\documentclass[pra,aps,showpacs,onecolumn,twoside,superscriptaddress]{revtex4}
\usepackage{amsmath,amsfonts,amssymb,caption,color,epsfig,graphics,graphicx,latexsym,mathrsfs,revsymb,theorem,url,verbatim,epstopdf}

\usepackage{hyperref}

\newtheorem{definition}{Definition}
\newtheorem{proposition}[definition]{Proposition}
\newtheorem{lemma}[definition]{Lemma}

\newtheorem{theorem}[definition]{Theorem}
\newtheorem{corollary}[definition]{Corollary}
\newtheorem{conjecture}[definition]{Conjecture}

\newtheorem{remark}[definition]{Remark}
\newtheorem{example}[definition]{Example}
\newtheorem{question}[definition]{Question}

\def\squareforqed{\hbox{\rlap{$\sqcap$}$\sqcup$}}
\def\qed{\ifmmode\squareforqed\else{\unskip\nobreak\hfil
\penalty50\hskip1em\null\nobreak\hfil\squareforqed
\parfillskip=0pt\finalhyphendemerits=0\endgraf}\fi}
\def\endenv{\ifmmode\;\else{\unskip\nobreak\hfil
\penalty50\hskip1em\null\nobreak\hfil\;
\parfillskip=0pt\finalhyphendemerits=0\endgraf}\fi}
\newenvironment{proof}{\noindent \textbf{{Proof.~} }}{\qed}
\def\Dbar{\leavevmode\lower.6ex\hbox to 0pt
{\hskip-.23ex\accent"16\hss}D}
\makeatletter
\def\url@leostyle{%
  \@ifundefined{selectfont}{\def\UrlFont{\sf}}{\def\UrlFont{\small\ttfamily}}}
\makeatother
\def\bcj{\begin{conjecture}}
\def\ecj{\end{conjecture}}
\def\bcr{\begin{corollary}}
\def\ecr{\end{corollary}}
\def\bd{\begin{definition}}
\def\ed{\end{definition}}
\def\bea{\begin{eqnarray}}
\def\eea{\end{eqnarray}}
\def\bem{\begin{enumerate}}
\def\eem{\end{enumerate}}
\def\bex{\begin{example}}
\def\eex{\end{example}}
\def\bim{\begin{itemize}}
\def\eim{\end{itemize}}
\def\bl{\begin{lemma}}
\def\el{\end{lemma}}
\def\bma{\begin{bmatrix}}
\def\ema{\end{bmatrix}}
\def\bpf{\begin{proof}}
\def\epf{\end{proof}}
\def\bpp{\begin{proposition}}
\def\epp{\end{proposition}}
\def\bqu{\begin{question}}
\def\equ{\end{question}}
\def\br{\begin{remark}}
\def\er{\end{remark}}
\def\bt{\begin{theorem}}
\def\et{\end{theorem}}

\def\btb{\begin{tabular}}
\def\etb{\end{tabular}}

\newcommand{\nc}{\newcommand}

\def\a{\alpha}
\def\b{\beta}
\def\g{\gamma}
\def\d{\delta}

\def\l{\lambda}
\def\m{\mu}
\def\n{\nu}

\def\p{\pi}
\def\r{\rho}
\def\s{\sigma}

\def\ph{\varphi}

\def\ps{\psi}

 \nc{\bbA}{\mathbb{A}} \nc{\bbB}{\mathbb{B}} \nc{\bbC}{\mathbb{C}}
 \nc{\bbD}{\mathbb{D}} \nc{\bbE}{\mathbb{E}} \nc{\bbF}{\mathbb{F}}
 \nc{\bbG}{\mathbb{G}} \nc{\bbH}{\mathbb{H}} \nc{\bbI}{\mathbb{I}}
 \nc{\bbJ}{\mathbb{J}} \nc{\bbK}{\mathbb{K}} \nc{\bbL}{\mathbb{L}}
 \nc{\bbM}{\mathbb{M}} \nc{\bbN}{\mathbb{N}} \nc{\bbO}{\mathbb{O}}
 \nc{\bbP}{\mathbb{P}} \nc{\bbQ}{\mathbb{Q}} \nc{\bbR}{\mathbb{R}}
 \nc{\bbS}{\mathbb{S}} \nc{\bbT}{\mathbb{T}} \nc{\bbU}{\mathbb{U}}
 \nc{\bbV}{\mathbb{V}} \nc{\bbW}{\mathbb{W}} \nc{\bbX}{\mathbb{X}}
 \nc{\bbZ}{\mathbb{Z}}

 \nc{\bA}{{\bf A}} \nc{\bB}{{\bf B}} \nc{\bC}{{\bf C}}
 \nc{\bD}{{\bf D}} \nc{\bE}{{\bf E}} \nc{\bF}{{\bf F}}
 \nc{\bG}{{\bf G}} \nc{\bH}{{\bf H}} \nc{\bI}{{\bf I}}
 \nc{\bJ}{{\bf J}} \nc{\bK}{{\bf K}} \nc{\bL}{{\bf L}}
 \nc{\bM}{{\bf M}} \nc{\bN}{{\bf N}} \nc{\bO}{{\bf O}}
 \nc{\bP}{{\bf P}} \nc{\bQ}{{\bf Q}} \nc{\bR}{{\bf R}}
 \nc{\bS}{{\bf S}} \nc{\bT}{{\bf T}} \nc{\bU}{{\bf U}}
 \nc{\bV}{{\bf V}} \nc{\bW}{{\bf W}} \nc{\bX}{{\bf X}}
 \nc{\bZ}{{\bf Z}}

\nc{\cA}{{\cal A}} \nc{\cB}{{\cal B}} \nc{\cC}{{\cal C}}
\nc{\cD}{{\cal D}} \nc{\cE}{{\cal E}} \nc{\cF}{{\cal F}}
\nc{\cG}{{\cal G}} \nc{\cH}{{\cal H}} \nc{\cI}{{\cal I}}
\nc{\cJ}{{\cal J}} \nc{\cK}{{\cal K}} \nc{\cL}{{\cal L}}
\nc{\cM}{{\cal M}} \nc{\cN}{{\cal N}} \nc{\cO}{{\cal O}}
\nc{\cP}{{\cal P}} \nc{\cQ}{{\cal Q}} \nc{\cR}{{\cal R}}
\nc{\cS}{{\cal S}} \nc{\cT}{{\cal T}} \nc{\cU}{{\cal U}}
\nc{\cV}{{\cal V}} \nc{\cW}{{\cal W}} \nc{\cX}{{\cal X}}
\nc{\cZ}{{\cal Z}}

\nc{\hA}{{\hat{A}}} \nc{\hB}{{\hat{B}}} \nc{\hC}{{\hat{C}}}
\nc{\hD}{{\hat{D}}} \nc{\hE}{{\hat{E}}} \nc{\hF}{{\hat{F}}}
\nc{\hG}{{\hat{G}}} \nc{\hH}{{\hat{H}}} \nc{\hI}{{\hat{I}}}
\nc{\hJ}{{\hat{J}}} \nc{\hK}{{\hat{K}}} \nc{\hL}{{\hat{L}}}
\nc{\hM}{{\hat{M}}} \nc{\hN}{{\hat{N}}} \nc{\hO}{{\hat{O}}}
\nc{\hP}{{\hat{P}}} \nc{\hR}{{\hat{R}}} \nc{\hS}{{\hat{S}}}
\nc{\hT}{{\hat{T}}} \nc{\hU}{{\hat{U}}} \nc{\hV}{{\hat{V}}}
\nc{\hW}{{\hat{W}}} \nc{\hX}{{\hat{X}}} \nc{\hZ}{{\hat{Z}}}

\nc{\hn}{{\hat{n}}}

\def\diag{\mathop{\rm diag}}
\def\dim{\mathop{\rm Dim}}

\def\max{\mathop{\rm max}}

\def\es{\emptyset}

\def\lra{\leftrightarrow}

\def\ox{\otimes}

\def\ra{\rightarrow}
\def\Ra{\Rightarrow}

\newcommand{\pp}[2]{{\partial #1\over\partial #2}}

\newcommand{\bra}[1]{\langle#1|}
\newcommand{\ket}[1]{|#1\rangle}
\newcommand{\proj}[1]{| #1\rangle\!\langle #1 |}

\newcommand{\braket}[2]{\langle#1|#2\rangle}

\newcommand{\tbc}{\red{TO BE CONTINUED...}}

\newcommand{\opp}{\red{OPEN PROBLEMS}.~}

\newcommand{\red}{\textcolor{red}}
\def\Dbar{\leavevmode\lower.6ex\hbox to 0pt
{\hskip-.23ex\accent"16\hss}D}
\begin{document}
\title{Constructing $2\times2\times4$ and $4\times4$ unextendible product bases and positive-partial-transpose entangled states}

\author{Lin Chen}\email[]{linchen@buaa.edu.cn (corresponding author)}
\affiliation{School of Mathematics and Systems Science, Beihang University, Beijing 100191, China}
\affiliation{International Research Institute for Multidisciplinary Science, Beihang University, Beijing 100191, China}

\author{Kai Wang}\email[]
{kaywong@buaa.edu.cn}
\affiliation{School of Mathematics and Systems Science, Beihang University, Beijing 100191, China}

\author{Yi Shen}
\affiliation{School of Mathematics and Systems Science, Beihang University, Beijing 100191, China}

\author{Yize Sun}
\affiliation{School of Mathematics and Systems Science, Beihang University, Beijing 100191, China}

\author{Lijun Zhao}
\affiliation{School of Mathematics and Systems Science, Beihang University, Beijing 100191, China}

\date{\today}


\begin{abstract}
The 4-qubit unextendible product basis (UPB) has been recently studied by [Johnston, J. Phys. A: Math. Theor. 47 (2014) 424034]. From this result we show that there is only one UPB of size $6$ and six UPBs of size $9$ in $\cH=\bbC^2\ox\bbC^2\ox\bbC^4$,  
three UPBs of size $9$ in $\cK=\bbC^4\ox\bbC^4$, and no UPB of size $7$ in $\cH$ and $\cK$. Furthermore we construct a 4-qubit positive-partial-transpose (PPT) entangled state $\r$ of rank seven, and show that it is also a PPT entangled state in $\cH$ and $\cK$, respectively. We analytically derive the geometric measure of entanglement of a special $\r$.
\end{abstract}

\maketitle

%\tableofcontents

\section{Introduction}
\label{sec:int}

In quantum information, the unextendible product bases (UPBs) have been found useful in various applications, such as the nonlocality without entanglement, the construction of positive-partial-transpose (PPT) entangled states and Bell inequalities 
without quantum violation \cite{AL01,Fen06,Chen2013The, dms03, Tura2012Four,Pittenger2003Unextendible, Terhal2012A, Szanto2016Complementary, Chen2014Unextendible,Hou2014A, Sollid2011Unextendible, Augusiak2012Tight, Dicarlo2010Preparation}. In particular, the multiqubit UPB have received much attentions \cite{Bravyi2004Unextendible,Reed2012Realization,Joh13,Johnston2014The, Chen2018Nonexistence}, since the multiqubit system is the mostly realizable system in experiments. First, 3-qubit UPBs have been constructed by \cite{Bravyi2004Unextendible}.
Second the 4-qubit UPBs have been fully classified assisted by programs \cite{Johnston2014The}, and recently the 4-qubit orthogonal basis has also been classified using a combinatorial idea \cite{1751-8121-50-39-395301}. Third, the $n$-qubit UPBs of cardinality $2^n-5$ has been proven non-existing \cite{Chen2018Nonexistence}, and it has solved an open problem in \cite{Johnston2014The}. The multiqubit UPBs have also been studied in terms of the formally orthogonal matrices and Hasse diagrams \cite{Chen2018Multiqubit, Chen2018The}. Nevertheless, so far the connection between multiqubit UPBs and multipartite UPBs in higher dimensions has been little studied. Understanding the connection helps construct more UPBs systematically using the known UPBs. This is the main motivation of this paper.

In this paper we apply the result of 4-qubit UPBs \cite{Johnston2014The} to construct UPBs in $\bbC^2\ox\bbC^2\ox\bbC^4$ and $\bbC^4\ox\bbC^4$. From this result we construct the UPBs in $\cH=\bbC^4\ox\bbC^2\ox\bbC^2$ and $\cK=\bbC^4\ox\bbC^4$ of size $6,7,8$ and $9$. We show that there is only one UPB of size $6$ in $\cH$, six UPBs of size $9$ in $\cH$ and three UPBs of size $9$ in $\cK$. To obtain our results on UPBs of size $9$, we shall review the so-called unextendible orthogonal matrices (UOM) that was firstly used for multiqubit UPBs in \cite{Chen2018Multiqubit}. We shall further construct the 4-qubit PPT entangled state $\r$ of rank seven in \eqref{eq:rho}. This is realized by The state constructed from the UPBs in \eqref{eq:size9,family11-1}. Then we investigate the geometric measure of entanglement of a special $\r$ in Theorem \ref{thm:gme}. As far as we know, such a state have been little studied due to the mathematical difficulty. Using \eqref{eq:rho}, we shall show that the state $\r$ is also a  PPT entangled state of rank seven in $\cH$ and $\cK$. Theorem \ref{thm:gme} also gives an upper bound of the geometric measure of entanglement of both the states in $\cH$ and $\cK$.

The rest of this paper is structured in the following way. In Sec. \ref{sec:pre} we introduce the notions and facts on the UPBs and the coarse graining. We investigate the coarse graining of 4-qubit UPBs of size $6,7$ in Sec. \ref{sec:size67}, and that of $9$ in Sec. \ref{sec:size9}, respectively. We present the application of our result in Sec. \ref{sec:ent}. We construct the 4-qubit PPT entangled state $\r$ of rank seven, and analytically derive the geometric measure of entanglement of a special $\r$. Finally we conclude in Sec. \ref{sec:con}.

\section{Preliminaries}
\label{sec:pre}

In this section we introduce the notions and facts used in this paper. First we review the notion of UPBs, introduce the two properties and equivalence of 
UPBs in Sec. \ref{subsec:qi}. Second we introduce the coarse graining in Sec. \ref{subsec:coarse}.

\subsection{unextendible product basis}
\label{subsec:qi}

In quantum mechanics, an $n$-partite quantum system of $A_1,A_2,...,A_n$ is characterized by an $n$-partite Hilbert space $\cH=\cH_1\ox\cdots\ox\cH_n$. We refer to the quantum state $\ket{\ps_i}\in\cH_i$ as a $\dim \cH_i$-dimensional vector.
The {\em product vector in $\cH$} is an $n$-partite nonzero vector of the form
$\ket{\psi_1}\ox\cdots\ox\ket{\psi_n}$. For simplicity it is written as $\ket{\psi_1,\ldots,\psi_n}$. For the convenience in mathematical arguments,
We do not distinguish product vectors that
are scalar multiples of each other. We say that a set of $n$-partite orthonormal product vectors $\{\ket{a_{i,1}},...,\ket{a_{i,n}}\}$ is an unextendible product basis (UPB) in $\cH$ if there is no $n$-partite product vector orthogonal to all vectors in the set. For example any orthonormal basis in $\cH$ is a UPB. It is trivial because its size $n=\dim\cH$. So we only consider UPBs with size smaller than $\dim\cH$. We shall refer to $\cH$ as an $n$-qubit space when $\dim\cH_i=2$ for all $j$. We will study 4-qubit space, and we refer to $A_1,A_2,A_3,A_4$ as $A,B,C,D$. Hence $\bbC^2\ox\bbC^2\ox\bbC^2\ox\bbC^2:= \cH_A\ox\cH_B\ox\cH_C\ox\cH_D$, and $\cH_{AB}:=\cH_A\ox\cH_B$,  etc. We refer to the vectors $\ket{0}:=\bma1\\0\ema,\ket{1}=\bma0\\1\ema$ as an orthonormal basis in $\bbC^2$. More generally we denote $\ket{x},\ket{x'}$ as another orthonormal basis in $\bbC^2$, where we may choose $x=a,b,c$, etc. The following two properties are clear from the definition of UPB.
\begin{itemize}
\item 
If $\{\ket{a_{i,1},...,a_{i,n}}\}_{i=1,...,s}$ is a UPB of size $s$ then so is $\{\ket{a_{i,\s(1)},...,a_{i,\s(n)}}\}_{i=1,...,s}$, where $\s$ is a permutation of the integers $1,2,...,n$. That is, if we switch arbitrarily the systems of a UPB then we obtain another UPB.
\item
If $\{\ket{a_{i,1},...,a_{i,n}}\}_{i=1,...,s}$ is a UPB of size $s$ then so is $\{U_1\ket{a_{i,1}}\ox...\ox U_n\ket{a_{i,n}}\}_{i=1,...,s}$, where $U_1,...,U_n$ are arbitrary unitary matrices. That is, performing any product unitary transformation $U_1\ox...\ox U_n$ on a UPB still produces a UPB. 
\end{itemize}
We shall use the above two properties throughout the paper. If we obtain a UPB from another by using the above two properties, then we say that the two UPBs are equivalent.

\subsection{coarse graining}
\label{subsec:coarse}

We say that the $m$-partite Hilbert space $\cH'$ is an $m$-partite coarse graining of $\cH$ if $\cH'=\cH'_1\ox ...\ox\cH'_m$, $m\ge2$ and $\cH'_j=\ox_{k\in S_j}\cH_k$ where $\cup^m_{j=1} S_j=\{1,...,n\}$ and $S_j\cap S_k=\es$, $\forall j,k$. So $m\le n$, and the $n$-partite $\cH'$ is exactly $\cH$. One can similarly define the coarse graining of a set of product vectors in $\cH$. That is, suppose the set of product vectors $\{\ket{a_{i,1},...,a_{i,n}}\}\in\cH$, $\ket{a_{i,1},...,a_{i,n}}=\ket{b_{i,1},...,b_{i,m}}$ and $\ket{b_{i,j}}\in\cH_j'$ for $j=1,..,m$. Then the set $\{\ket{b_{i,1},...,b_{i,m}}\}$ consists of product vectors 
$\ket{b_{i,1},...,b_{i,m}}\in\cH'$. The set is defined as a coarse graining of the set $
\{\ket{a_{i,1},...,a_{i,n}}\}$.
We present the following claim.
\begin{lemma}
\label{le:n,n-1}
If a set of $n$-partite product vectors is a UPB in some coarse graining of the space $\cH$, then the set is a UPB in $\cH$.
\end{lemma}
\begin{proof}
If the set of orthogonal product vectors  $\cS=\{\ket{v_j}\}$ is not a UPB in $\cH$, then there exists a nonzero product vector $\ket{w}\in\cH$ such that $\braket {v_j}w =0$ for all $\ket{v_j}\in\cS$. By definition, the coarse graining of $\cS$ is still not a UPB. The converse is wrong, and we will see that some four-qubit UPB is no longer a UPB, and some others remain a UPB in the coarse graining $\bbC^4\ox\bbC^2\ox\bbC^2$.	
\end{proof}

Let $\cU_{2,2,2,2}$ be the set of 4-qubit UPBs. They have been fully characterized in \cite{Johnston2014The}. 
In this paper we investigate the UPBs in the coarse graining of 4-qubit UPBs.
First we define the subset $\cU_{2,2,4}\subseteq\cU_{2,2,2,2}$,   where any element of $\cU_{2,2,4}$ is a UPB in one of the six spaces $\cH_A\ox\cH_B\ox\cH_{CD}$, $\cH_A\ox\cH_C\ox\cH_{BD}$, $\cH_A\ox\cH_D\ox\cH_{BC}$, $\cH_B\ox\cH_C\ox\cH_{AD}$, $\cH_B\ox\cH_D\ox\cH_{AC}$, and $\cH_C\ox\cH_D\ox\cH_{AB}$. 
Second
one can similarly define the subset $\cU_{4,4}\subseteq\cU_{2,2,2,2}$, where any element of $\cU_{4,4}$ is a UPB in one of the three spaces $\cH_{AB}\ox\cH_{CD}$, $\cH_{AC}\ox\cH_{BD}$, and $\cH_{AD}\ox\cH_{BC}$. These definitions and Lemma \ref{le:n,n-1} imply the relation
\begin{eqnarray}
\label{eq:inclusion}
\cU_{2,2,2,2}\supseteq	
\cU_{2,2,4}\supseteq	
\cU_{4,4}.
\end{eqnarray} 
It is known that any 4-qubit UPB has size $s=6,7,8,9,10$ or $12$ \cite{Johnston2014The}. We can split the sets $\cU_{2,2,2,2}$, $\cU_{2,2,4}$ and $\cU_{4,4}$ into disjoint subsets $\cU_{2,2,2,2}^s$, $\cU_{2,2,4}^s$ and $\cU_{4,4}^s$ consisting of UPBs of size $s$, respectively. So we have
\begin{eqnarray}
\label{eq:u2222}
&&
\cU_{2,2,2,2}
=
\cup_{s=6,7,8,9,10,12}
\quad
\cU_{2,2,2,2}^s,
\\&&
\label{eq:u224}
\cU_{2,2,4}
=
\cup_{s=6,7,8,9,10,12}
\quad 
\cU_{2,2,4}^s,
\\&&
\label{eq:u44}
\cU_{4,4}
=
\cup_{s=6,7,8,9,10,12}
\quad 
\cU_{4,4}^s,
\\&&
\label{eq:inclusion=s}
\cU_{2,2,2,2}^s\supseteq	
\cU_{2,2,4}^s\supseteq	
\cU_{4,4}^s,
\quad\quad\quad
s=6,7,8,9,10,12.
\end{eqnarray}

For example, we shall show in Proposition \ref{pp:size6} that for $4$-qubit UPBs of size $6$ the first inclusion in \eqref{eq:inclusion=s} is strict, and there is no UPB in the coarse graining $\cH_{AB}\ox\cH_{CD}$, $\cH_{AC}\ox\cH_{BD}$, and $\cH_{AD}\ox\cH_{BC}$. That is $\cU_{2,2,2,2}^s\supset	
\cU_{2,2,4}^s$ and $	
\cU_{4,4}^s=\es$. In this paper we shall characterize the sets $\cU_{2,2,4}$, $\cU_{4,4}$, $\cU_{2,2,4}^s$, and $\cU_{4,4}^s$.

If we denote $\cS_{A:B:C:D}\in\cU^s_{2,2,2,2}$ as a $4$-qubit UPB of size $s$, then we denote $\cS_{A:B:CD}$, $\cS_{A:C:BD}$, $\cS_{A:D:BC}$, $\cS_{B:C:AD}$, $\cS_{B:D:AC}$ and $\cS_{C:D:AB}$ as the corresponding sets of product vectors in the above-mentioned six spaces $\cH_A\ox\cH_B\ox\cH_{CD}$, $\cH_A\ox\cH_C\ox\cH_{BD}$, $\cH_A\ox\cH_D\ox\cH_{BC}$, $\cH_B\ox\cH_C\ox\cH_{AD}$, $\cH_B\ox\cH_D\ox\cH_{AC}$, and $\cH_C\ox\cH_D\ox\cH_{AB}$, respectively.
We call them the coarse graining of $\cS_{A:B:C:D}$. One can similarly define the coarse graining $\cS_{AB:CD}$, $\cS_{AC:BD}$, and $\cS_{AD:BC}$. If $\cS_{A:B:CD},...,\cS_{AD:BC}$ are indeed UPBs of size $s$, then they are the elements in $\cU^s_{2,2,4}$ and $\cU^s_{4,4}$, respectively.
It follows from Lemma \ref{le:n,n-1} that if $\cS_{AB:CD}$ is a UPB then so are $\cS_{A:B:CD}$ and $\cS_{C:D:AB}$, though the converse may fail. That is
\begin{eqnarray}
\label{eq:abcd}
&&
\forall 
\cS_{A:B:C:D}\in\cU^s_{2,2,2,2},
\\&&
\cS_{AB:CD}\in \cU^s_{4,4}
\quad\Ra\quad	
\cS_{A:B:CD},\cS_{C:D:AB}\in \cU^s_{2,2,4},
\notag\\&&
\cS_{A:B:CD},\cS_{C:D:AB}\in \cU^s_{2,2,4}
\quad\not\Ra\quad
\cS_{AB:CD}\in \cU^s_{4,4}.	
\notag
\end{eqnarray}

On the other hand, it is known that any set of orthogonal product vectors of size smaller than $2d$ in $\bbC^d\ox\bbC^2$ is not a UPB \cite{BDM+99}. So 
\begin{lemma}
\label{le:space=82}	
Any $4$-qubit UPB is no longer a UPB in the coarse graining $\bbC^8\ox\bbC^2$.
\end{lemma}
So it suffices to consider only the coarse graining $\bbC^2\otimes\bbC^2\otimes\bbC^4$ and $\bbC^4\otimes\bbC^4$, as we shall do from the next section. As the final remark of this section, the following observation will be used in the proof of Lemma \ref{le:size9,family1-10}.

\begin{lemma}
\label{le:n-3of2*2*4}
Suppose that $\cS_{A:B:C:D}$=\{$\ket{f_1,g_1,h_1,i_1}$, $\ket{f_2,g_2,h_2,i_2}$,...,$\ket{f_n,g_n,h_n,i_n}$\} is a 4-qubit UPB of size n. Suppose $m$ is an integer satisfying $0\leq m\le n-4$. In the coarse graining $\bbC^2\ox\bbC^2\ox\bbC^4$, the set $\cS_{A:B:CD}$ is no longer a UPB if it satisfies the following two conditions. 

(i) There exists an integer $k\leq n-3-m$, such that $\ket{f_1}=\ket{f_2}=...=\ket{f_k}$ and $\ket{g_{k+1}}=...=\ket{g_{n-3-m}}$.

(ii) There are $m+1$ pairwise linearly dependent product vectors in the set $\{\ket{h_{n-2-m},i_{n-2-m}},...,\ket{h_n,i_n}\}$.

\end{lemma}

\begin{proof}
  The space spanned by $\ket{h_{n-2-m},i_{n-2-m}}$, $\ket{h_{n-1-m},i_{n-1-m}}$,...,$\ket{h_n,i_n}$ has dimension at most three.  That is, there exists a product vector $\ket{\ph}\in\bbC^4$ orthogonal to $\ket{h_{n-2-m},i_{n-2-m}}$, $\ket{h_{n-1-m},i_{n-1-m}}$,...,$\ket{h_n,i_n}$. Consequently, the product vector $\ket{f_1',g_{k+1}',\ph}$ is orthogonal to $\cS_{A:B:CD}$.
\end{proof}

\section{The coarse graining of $4$-qubit UPBs of size $6$ and $7$}
\label{sec:size67}

It is known that 4-qubit UPBs have size at least $6$. We begin with the simplest case, namely 4-qubit UPBs of size $6$. The following set $\cS_{A:B:C:D}$ is the only $4$-qubit UPB of size $6$ in $\cH_A\ox\cH_B\ox\cH_C\ox\cH_D$ \cite[Table 1]{Johnston2014The}.
\begin{eqnarray}
\label{eq:0000}
\ket{0,0,0,0},
\ket{0,a,a,1},
\ket{1,0,b,a},
\ket{1,a,b',b},
\ket{a,1,a',b'},
\ket{a',a',1,a'}.
\end{eqnarray}
Let $\cS_{A:B:CD}$ be a coarse graining of $\cS_{A:B:C:D}$, in the sense that $\cS_{A:B:CD}$ consists of product vectors in $\cH_A\ox\cH_B\ox\cH_{CD}$, where $\cH_{CD}=\bbC^4$. Although they are still orthogonal to each other,
one can show that $\cS_{A:B:CD}$ is not a UPB as follows. 
Let $\ket{\ps} \in{\cH_{CD}}$ be a nonzero vector
such that $\ket{\ps}$ is orthogonal to $\ket{b',b}$, $\ket{a',b'}$, and $\ket{1,a'}$. Then one can verify that $\ket{1,1,\ps}$ is orthogonal to $\cS_{A:B:CD}$.

Using the similar argument, one can show that none of the coarse graining $\cS_{A:BC:D}$, $\cS_{A:BD:C}$, $\cS_{AC:B:D}$, and $\cS_{AD:B:C}$ is a UPB.
For $\ps_1\in \cH_{BC}$ orthhogonal to $\ket{a,b'}, \ket{1,a'},$ and $ \ket{a',1}$, we have $\ket{1,\ps_1,a'}$ is orthogonal to $\cS_{A:BC:D}$. For $\ps_2\in \cH_{BD}$ orthhogonal to $\ket{a,b},\ket{1,b'},$ and $\ket{a',a'}$, we have $\ket{1,\ps_2,b'}$ is orthogonal to $\cS_{A:BD:C}$. For $\ps_3\in \cH_{AC}$ orthhogonal to $\ket{1,b'},\ket{a,a'},$ and $\ket{a',1}$, we have $\ket{\ps_3,1,0}$ is orthogonal to $\cS_{AC:B:D}$. For $\ps_4\in \cH_{AD}$ orthhogonal to $\ket{1,b},\ket{a,b'},$ and $\ket{a',a'}$, $\ket{\ps_4,1,a'}$ is orthogonal to $\cS_{AD:B:C}$.
We present more facts as follows.

\begin{proposition}
\label{pp:size6}
If $\cS_{A:B:C:D}$ is the 4-qubit UPB of size $6$, then

(i) $\cS_{AB:C:D}$ is a UPB. 

(ii) $\cS_{A:B:C:D}$ and $\cS_{AB:C:D}$ are the only two UPBs in all coarse graining of $\cH$. That is,
\begin{eqnarray}
\cU_{2,2,4}^6=\{\cS_{AB:C:D}\},
\quad
\cU_{4,4}^6=\es.	
\end{eqnarray}
\end{proposition}
\begin{proof} 
(i) Evidently, every $\ket{v}\in \cS_{AB:C:D}$ is a product vector and $\braket vw=0$ for all $\ket{v}\neq\ket{w}\in \cS_{AB:C:D}$. 
Suppose there exists a nonzero product vector $\ket{\ps,c,d}_{AB:C:D}$ orthogonal to $\cS_{AB:C:D}$. We note that any two of $\ket{0},\ket{1},\ket{a},\ket{a'},\ket{b},\ket{b'}$ are linearly independent. So $\ket{c,d}_{CD}$ is orthogonal to at most two vectors in $\cS_{AB:C:D}$. 
It implies that $\ket{\ps}_{AB}$ is orthogonal to at least four vectors in $\cS_{AB:C:D}$. It is a contradiction with the observation that any four of $\ket{0,0},\ket{0,a},\ket{1,0},\ket{1,a},\ket{a,1},\ket{a',a'}$ are linearly independent. 

(ii) It follows from (i) and the discussion below \eqref{eq:0000}  that $\cS_{AB:C:D}$ is the only UPB in all 3-partite coarse grainings of $\cH$. Next we consider $\cS$ in bipartite coarse grainings $\cH'$ of $\cH$. The latter is $\bbC^4\ox \bbC^4$ or $\bbC^8\ox\bbC^2$. If $\cH'=\bbC^4\ox \bbC^4$ then $\cS$  becomes $\cS_{AB:CD}$, $\cS_{AC:BD}$ or $\cS_{AD:BC}$. If $\cH'=\bbC^2\ox \bbC^8$ then $\cS$ becomes $\cS_{ABC:D}$, $\cS_{ABD:C}$, $\cS_{ACD:B}$ or $\cS_{BCD:A}$. None of them is a UPB in terms of Lemma \ref{le:n,n-1}, and the fact that $\cS_{AB:C:D}$ is the only UPB in all 3-partite coarse grainings of $\cH$.
\end{proof}

Using arguments similar to those for UPBs of size $6$, we investigate the UPBs of size $7$.
\begin{proposition}
\label{pp:size7}	
The $4$-qubit UPBs of size $7$ are the only UPBs in all coarse graining of $\cH$. That is,
\begin{eqnarray}
\label{eq:size7}	
\cU_{2,2,4}^7=\cU_{4,4}^7=\es.
\end{eqnarray}
\end{proposition}
\begin{proof}
We shall prove $\cU_{2,2,4}^7=\es.$
Then \eqref{eq:inclusion=s} implies the claim. 

Take the set \{$\ket{0,0,0,0}$, $\ket{0,a,a,1}$, $\ket{0,a',1,a}$, $\ket{1,0,0,b}$, $\ket{1,a',a,b'}$, $\ket{a,a,1,0}$, $\ket{a',1,a',a'}$\}\cite[Table 1]{Johnston2014The} as $S_{A:B:C:D}$. We found that the vector $\ket{0}$ in the first qubit has multipicity three. Moreover, the space spanned by $\ket{a,b'}$, $\ket{1,0}$ and $\ket{a',a'}$ has diemnsion at most three. That is, there is $\ket{p}\in\bbC^4$ orthogonal to $\ket{a,b'}$, $\ket{1,0}$ and $\ket{a',a'}$. So $\ket{1,1,p}$ is orthogonal to $\cS_{A:B:CD}$. Similarly, the set $\cS_{A:C:BD}$ and $\cS_{A:D:BC}$ are no longer UPBs. For $\cS_{AB:C:D}$, there are $\ket{a}$ of multiplicity two in the third qubit and $\ket{0}$ of multiplicity two in the fourth qubit. Moreover, the two $\ket{a}$ and two $\ket{0}$ are in four distinct product vectors. So there exists $\ket{q}\in\bbC^4$ orthogonal to the last two qubits of the rest product vectors of $\cS_{A:B:C:D}$. So $\ket{q,a',1}$ is orthogonal to $\cS_{AB:C:D}$. Similarly, $\cS_{AC:B:D}$ and $\cS_{AD:B:C}$ are also no longer UPBs.
\end{proof}

\section{The coarse graining of $4$-qubit UPBs of size $9$}
\label{sec:size9}
In this section we investigate 
the coarse graining of $4$-qubit UPBs of size $9$. We show in Proposition \ref{pp:size9} that there are six UPBs of size $9$ in $\bbC^2\ox\bbC^2\ox\bbC^4$,  
and three UPBs of size $9$ in $\bbC^4\ox\bbC^4$. For this purpose we introduce the unextendible orthogonal matrices (UOM) \cite[p1]{Chen2018Multiqubit}. We take product vectors of an $n$-partite UPB of size $m$ as row vectors of an $m\times n$ matrix, so that the matrix is known as the UOM. For example, the three-qubit UPB $\ket{0,0,0},\ket{1,+,-},\ket{-,1,+}$ and $\ket{+,-,1}$ can be expressed as the UOM
\begin{eqnarray}
\bma
0&0&0\\
1&+&-\\
-&1&+\\
+&-&1\\
\ema.	
\end{eqnarray}
Using the orthogonality of product vectors in the UPB, we shall simply say that the rows of UOM are orthogonal. For orthogonal states $\ket{a}$ and $\ket{a'}$ we shall refer to them as $a$ and $a'$ in UOMs, and vice versa.

We start by presenting two preliminary lemmas. 

\begin{lemma}
\label{le:size9,family1-10}	
For the first ten UPBs of size $9$ in \cite[Table 1]{Johnston2014The}, neither of them is a UPB in the coarse graining $\bbC^2\ox\bbC^2\ox\bbC^4$ or $\bbC^4\ox\bbC^4$. 
\end{lemma}
\begin{proof}
Suppose $\cS_{A:B:C:D}$=\{$\ket{f_1,g_1,h_1,i_1}$, $\ket{f_2,g_2,h_2,i_2}$,...,$\ket{f_9,g_9,h_9,i_9}$\} is a 4-qubit UPB of size 9. We write the UOMs of the first ten UPBs of size 9 as the matrices $U_1, U_2,...,U_{10}$. By observation, one can find that all these matrices can be classified into three categories.
In the first one, there are the two columns, one of which has four identical vectors and another has two identical vectors in remaining rows.  
In the second one, there are the two columns, one of which has three identical vectors and another has three identical vectors in remaining rows. 
 In the last one, there are the two columns, one of which has three identical vectors and another has two identical vectors in remaining rows. Besides, there are two linearly dependent vectors in other rows of the remaining two columns.

  Up to equivalence of UPBs, the three categories matrices respectively has the same structure as the following three matrices in \eqref{eq:10UPBsofsize9}. 

\begin{eqnarray}
\label{eq:10UPBsofsize9}
\bma
0&*&*&*\\
0&*&*&*\\
0&*&*&*\\
0&*&*&*\\
*&0&*&*\\
*&0&*&*\\
*&*&*&*\\
*&*&*&*\\
*&*&*&*\\
\ema,
\quad
\bma
0&*&*&*\\
0&*&*&*\\
0&*&*&*\\
*&0&*&*\\
*&0&*&*\\
*&0&*&*\\
*&*&*&*\\
*&*&*&*\\
*&*&*&*\\
\ema,
\quad
\bma
0&*&*&*\\
0&*&*&*\\
0&*&*&*\\
*&0&*&*\\
*&0&*&*\\
*&*&0&d\\
*&*&0&d\\
*&*&*&*\\
*&*&*&*\\
\ema
\end{eqnarray}

 Considering $U_1$, we have $\ket{f_1}=\ket{f_5}=\ket{f_6}=\ket{f_8}=\ket{0}$ and $\ket{g_3}=\ket{g_7}=\ket{a'}$ and m=0 in lemma \ref{le:n-3of2*2*4}. So $\cS_{A:B:CD}$ is no longer a UPB from lemma \ref{le:n-3of2*2*4}. Using arguments similar to $U_1$, one can find that neither of these UPBs is a UPB in the coarse graining $\bbC^2\ox\bbC^2\ox\bbC^4$. Certainly, neither of them is a UPB in the coarse graining $\bbC^4\ox\bbC^4$ from lemma \ref{le:n,n-1}.

\end{proof}

\begin{lemma}
\label{le:size9,family11}
(i) Suppose $\cS_{A:B:C:D}$ is the 11'th UPB of size $9$ in \cite[Table 1]{Johnston2014The}. Then $\cS_{AB:CD}$, $\cS_{AC:BD}$ and $\cS_{AD:BC}$ are  simultaneously UPBs or not.

(ii) $\cS_{AB:CD}$ and $\cS_{CD:BA}$ are the same up to row permutation and product unitary transformation. 
\end{lemma}
\begin{proof}
(i) It suffices to show that $\cS_{AB:CD}$, $\cS_{AC:BD}$ and $\cS_{AD:BC}$ are equivalent. 

We write the UOM of 11'th UPB $U$ of size $9$ in \cite[Table 1]{Johnston2014The} as the first matrix $U_1$ in \eqref{eq:size9,family11}. The remaining matrices in \eqref{eq:size9,family11} are respectively denoted as $U_2,U_3$ and $U_4$. 
We convert $U_1$ into $U_2$ by  switching $0$ and $1$ in the third column of $U_1$, and $(0,1)\lra (a,a')$ in the fourth column of $U_1$. They can be realized by performing product unitary transformation on the third and fourth qubit of the UPB $U$. Next, we convert $U_2$ into $U_3$ by switching column $2,3,4$ of $U_2$ into its column $4,2,3$.
Finally, we obtain $U_4$ by permuting row $1,2,3$ of $U_3$, and permuting row $4,5,6$ of $U_3$, respectively. 

So $U_2$ and $U_4$ are the same. The switching in the last paragraph shows that $\cS_{AB:CD}$ and $\cS_{AC:DB}$ are  simultaneously UPBs or not. One can similarly prove that $\cS_{AB:CD}$ and $\cS_{AD:BC}$ are  simultaneously UPBs or not. 

We have shown that $\cS_{AB:CD}$, $\cS_{AC:BD}$ and $\cS_{AD:BC}$ are equivalent. So the assertion holds.

\begin{eqnarray}
\label{eq:size9,family11}
\bma
0&0&0&0\\
0&1&a&a\\
0&a&1&a'\\
1&1&1&0\\
1&a&0&a\\
1&0&a&a'\\
a&0&1&a\\
a&1&0&a'\\
a'&a'&a'&1\\	
\ema
\ra 
\bma
0&0&1&a\\
0&1&a&0\\
0&a&0&1\\
1&1&0&a\\
1&a&1&0\\
1&0&a&1\\
a&0&0&0\\
a&1&1&1\\
a'&a'&a'&a'\\	
\ema
\ra\bma
0&1&a&0\\
0&a&0&1\\
0&0&1&a\\
1&0&a&1\\
1&1&0&a\\
1&a&1&0\\
a&0&0&0\\
a&1&1&1\\
a'&a'&a'&a'\\	
\ema
\ra 
\bma
0&0&1&a\\
0&1&a&0\\
0&a&0&1\\
1&1&0&a\\
1&a&1&0\\
1&0&a&1\\
a&0&0&0\\
a&1&1&1\\
a'&a'&a'&a'\\	
\ema.
\end{eqnarray}

(ii) We put down the last matrix in \eqref{eq:size9,family11} as the first matrix in \eqref{eq:size9,family11-1}. We name the matrices in \eqref{eq:size9,family11-1} as $U_1,...,U_6$, respectively. We obtain $U_2$ by switching columns $1,2$ and $3,4$ of $U_1$.
We obtain $U_3$ by switching the rows of $U_2$. We obtain $U_4$ by switching column $3$ and $4$ of $U_3$. We obtain $U_5$ by switching the symbols $0$ and $1$ in column $3$ of $U_4$. Finally we obtain $U_6$ by switching row $4,5,6$ of $U_5$. 
\begin{eqnarray}
\label{eq:size9,family11-1}
\bma
0&0&1&a\\
0&1&a&0\\
0&a&0&1\\
1&1&0&a\\
1&a&1&0\\
1&0&a&1\\
a&0&0&0\\
a&1&1&1\\
a'&a'&a'&a'\\	
\ema
\ra	
\bma
1&a&0&0\\
a&0&0&1\\
0&1&0&a\\
0&a&1&1\\
1&0&1&a\\
a&1&1&0\\
0&0&a&0\\
1&1&a&1\\
a'&a'&a'&a'\\	
\ema
\ra	
\bma
0&0&a&0\\
0&1&0&a\\
0&a&1&1\\
1&0&1&a\\
1&a&0&0\\
1&1&a&1\\
a&0&0&1\\
a&1&1&0\\
a'&a'&a'&a'\\	
\ema
\ra	
\bma
0&0&0&a\\
0&1&a&0\\
0&a&1&1\\
1&0&a&1\\
1&a&0&0\\
1&1&1&a\\
a&0&1&0\\
a&1&0&1\\
a'&a'&a'&a'\\	
\ema
\ra	
\bma
0&0&1&a\\
0&1&a&0\\
0&a&0&1\\
1&0&a&1\\
1&a&1&0\\
1&1&0&a\\
a&0&0&0\\
a&1&1&1\\
a'&a'&a'&a'\\	
\ema
\ra
\bma
0&0&1&a\\
0&1&a&0\\
0&a&0&1\\
1&1&0&a\\
1&a&1&0\\
1&0&a&1\\
a&0&0&0\\
a&1&1&1\\
a'&a'&a'&a'\\	
\ema.	
\end{eqnarray}
\end{proof}

Based on the above two lemmas, we present the main result of this section.

\begin{proposition}
\label{pp:size9}
Suppose $\cS_{A:B:C:D}$ is the 11'th UPB of size $9$ in \cite[Table 1]{Johnston2014The}, and its UOM is the first matrix in \eqref{eq:size9,family11-1}. Then 
\begin{eqnarray}
\label{eq:siza9=224}	
&&
\cU_{2,2,4}^9=\{\cS_{A:B:CD}, \cS_{A:C:BD}, \cS_{A:D:BC}, \cS_{B:C:AD}, \cS_{B:D:AC}, \cS_{C:D:AB}\},
\\&&\label{eq:siza9=44}
\cU_{4,4}^9=\{\cS_{AB:CD},\cS_{AC:BD},\cS_{AD:BC}\}.
\end{eqnarray}
\end{proposition}
\begin{proof}
We claim that $\cU_{4,4}^9\supseteq \{\cS_{AB:CD},\cS_{AC:BD},\cS_{AD:BC}\}$.
Recall that there are exactly $11$ UPBs of size $9$ by  \cite[Table 1]{Johnston2014The}. So
Lemma \ref{le:size9,family1-10}	 implies that \eqref{eq:siza9=44} holds. 
Then \eqref{eq:siza9=224} holds by \eqref{eq:abcd} and Lemma \ref{le:size9,family1-10}.

In the following we prove the claim.
Let $\cT_{A:B:C:D}$=\{$\ket{f_1,g_1,h_1,i_1}$, $\ket{f_2,g_2,h_2,i_2}$,...,$\ket{f_9,g_9,h_9,i_9}$\} be the 11'th UPB of size $9$. From $U_1$, one can show that any six product vectors in the set $\{\ket{f_1,g_1},\ket{f_2,g_2},...,\ket{f_9,g_9}\}$ span a space $\bbC^4$, and any six product vectors in the set $\{\ket{h_1,i_1},\ket{h_2,i_2},...,\ket{h_9,i_9}\}$ span a space $\bbC^4$.

Suppose $\cT_{AB:CD}$ is not a UPB. Then there is a vector $\ket{\a,\b}\in\cH_{AB}\ox\cH_{CD}$ orthogonal to $\cT_{AB:CD}$. Up to the permutation of subscripts, we can assume that $\ket{\a}$ is orthogonal to $\{\ket{f_1,g_1},\ket{f_2,g_2},...,\ket{f_m,g_m}\}$, and $\ket{\b}$ is orthogonal to $\{\ket{h_{m+1},i_{m+1}},\ket{h_{m+2},i_{m+2}},...,\ket{h_9,i_9}\}$.
The fact in the last paragraph shows that $m=4$ or $5$. Using Lemma \ref{le:size9,family11}, we only need to prove the assertion for $m=5$.

We write $\ket{f_1,g_1},\ket{f_2,g_2},...,\ket{f_9,g_9}$ as the 4-dimensional vector $u_1$,$u_2$,...,$u_9$. One can verify that $u_1$, $u_2$, $u_3$ span a space $\bbC^2$, so do $u_4$, $u_5$, $u_6$. Let $\{T_i\}_{i=1}^{\binom{9}{5}}$ be the collections of any five vectors among $u_1$, $u_2$,...,$u_9$. In order to judge linear dependence of the vectors in $T_i$, we divide all $T_i$ for $i=1,...,\binom{9}{5}$ into three classes $\{T_i\}_{\land_1}$, $\{T_i\}_{\land_2}$, $\{T_i\}_{\land_3}$. If $T_i$ includes $u_1$, $u_2$, $u_3$ (or $u_4$, $u_5$, $u_6$) and at most two of $u_4$, $u_5$, $u_6$ (or $u_1$, $u_2$, $u_3$), then $T_i\in\{T_i\}_{\land_1}$. If $T_i$ includes two of $u_1$, $u_2$, $u_3$ (or $u_4$, $u_5$, $u_6$) and at most two of $u_4$, $u_5$ $u_6$ (or $u_1$, $u_2$, $u_3$), then $T_i\in\{T_i\}_{\land_2}$. If $T_i$ includes one of $u_1$, $u_2$, $u_3$ (or $u_4$, $u_5$, $u_6$) and one of $u_4$, $u_5$ $u_6$ (or $u_1$, $u_2$, $u_3$), then $T_i\in\{T_i\}_{\land_3}$.
 
For $T_i\in \{T_i\}_{\land_1}$, there are three cases, the three vectors $u_1$, $u_2$, $u_3$ (or $u_4$, $u_5$, $u_6$) and two of $u_7$, $u_8$, $u_9$, the three $u_1$, $u_2$, $u_3$ (or $u_4$, $u_5$, $u_6$) and two of  $u_4$, $u_5$, $u_6$ (or $u_1$, $u_2$, $u_3$),  the three $u_1$, $u_2$, $u_3$ (or $u_4$, $u_5$, $u_6$) and one of  $u_4$, $u_5$, $u_6$ (or $u_1$, $u_2$, $u_3$) and one of $u_7$, $u_8$, $u_9$. One can verify that $T_i$ in the first two cases span a space $\bbC^4$. In the last case, the four sets $T_i=\{u_1, u_2, u_3, u_4, u_8\}$, $T_j=\{u_1, u_2, u_3, u_6, u_7\}$, $T_k=\{u_1, u_4, u_5, u_6, u_7\}$, $T_l=\{u_2, u_4, u_5, u_6, u_8\}$ all span a space $\bbC^3$ and others span a space $\bbC^4$.

For $T_i\in \{T_i\}_{\land_2}$, there are three cases, two of $u_1$, $u_2$, $u_3$ (or $u_4$, $u_5$, $u_6$) and three of $u_7$, $u_8$, $u_9$, two of $u_1$, $u_2$, $u_3$ (or $u_4$, $u_5$, $u_6$) and one of  $u_4$, $u_5$, $u_6$ (or $u_1$, $u_2$, $u_3$) and two of $u_7$, $u_8$, $u_9$, two of $u_1$, $u_2$, $u_3$ (or $u_4$, $u_5$, $u_6$) and two of  $u_4$, $u_5$, $u_6$ (or $u_1$, $u_2$, $u_3$) and one of $u_7$, $u_8$, $u_9$. One can verify that $T_i$ span a space $\bbC^4$. 

For $T_i\in \{T_i\}_{\land_3}$, there are only a case, one of $u_1$, $u_2$, $u_3$ (or $u_4$, $u_5$, $u_6$) and one of  $u_4$, $u_5$, $u_6$ (or $u_1$, $u_2$, $u_3$) and the three $u_7$, $u_8$, $u_9$. Evidently, $T_i$ span a space $\bbC^4$.

So we need to investigate $T_i$, $T_j$, $T_k$, $T_l$. We express them as the submatrices of corresponding UOMs 
\begin{eqnarray}
\bma
0&0\\
0&1\\
0&a\\
1&1\\
a&1\\	
\ema,
\quad	
\bma
0&0\\
0&1\\
0&a\\
1&0\\
a&0\\	
\ema,
\quad
\bma
0&0\\
1&1\\
1&a\\
1&0\\
a&0\\
\ema,
\quad	
\bma
0&1\\
1&1\\
1&a\\
1&0\\
a&1\\	
\ema.
\end{eqnarray}
Their cofactors are respectively 
\begin{eqnarray}
\label{eq:non}
\bma
0&a\\
a&a'\\
1&a\\
a'&1\\	
\ema,
\quad	
\bma
1&0\\
0&a\\
0&a'\\
a'&1\\	
\ema,
\quad
\bma
a&a\\
1&a'\\
0&a'\\
a'&1\\
\ema,
\quad	
\bma
0&0\\
1&a'\\
1&a\\
a'&1\\	
\ema.
\end{eqnarray}
 By observing each of the four cases, we can obtain that there is no product vector in $\cH_{AB}:\cH_{CD}$ orthogonal to $\cU^9_{4,4}$.
\end{proof}

\section{The construction and entanglement of 4-qubit positive-partal-transpose entangled states} 
\label{sec:ent}

In this section we present two main results as the application of the previous section. First we construct the 4-qubit PPT entangled state $\r$ of rank seven in \eqref{eq:rho}. Second we investigate the geometric measure of entanglement of a special $\r$ in Theorem \ref{thm:gme}. As far as we know, such a state have been little studied due to the mathematical difficulty. Using \eqref{eq:rho}, we shall show that the state $\r$ is also a $2\times2\times4$ and $4\times4$ positive-partial-transpose (PPT) entangled state of rank seven in terms the partition of systems $A:B:CD$ and $AB:CD$. The state is constructed from the UPBs by UOMs in \eqref{eq:size9,family11-1} in the last section.

Recall that the vectors in different columns of the first matrix in \eqref{eq:size9,family11-1} are different, though we name all of them as $a,a'$ for convenience. To distinguish them in the UPBs, we rename them as $\ket{a},\ket{a'}$, $\ket{b},\ket{b'}$, $\ket{c},\ket{c'}$ and $\ket{d},\ket{d'}$, respectively. Their general expressions are in this form $\cos\a\ket{0}+e^{i\theta}\sin\a\ket{1}$. However we can simplify it by performing a diagonal unitary matrix $\diag(1,e^{-i\theta})$ on the above qubits. Then $\cos\a\ket{0}+e^{i\theta}\sin\a\ket{1}$ becomes $\cos\a\ket{0}+\sin\a\ket{1}$. We still name them as $\ket{a},\ket{a'}$, $\ket{b},\ket{b'}$, $\ket{c},\ket{c'}$ and $\ket{d},\ket{d'}$, respectively. They are all orthonormal basis in $\bbC^2$ with the following expressions.
\begin{eqnarray}
\label{eq:abcd}	
&&
\ket{a}=\cos\a\ket{0}+\sin\a\ket{1},
\quad\quad\quad
\ket{a'}=\sin\a\ket{0}-\cos\a\ket{1},
\notag\\&&
\ket{b}=\cos\b\ket{0}+\sin\b\ket{1},
\quad\quad\quad
\ket{b'}=\sin\b\ket{0}-\cos\b\ket{1},
\notag\\&&
\ket{c}=\cos\g\ket{0}+\sin\g\ket{1},
\quad\quad\quad
\ket{c'}=\sin\g\ket{0}-\cos\g\ket{1},
\notag\\&&
\ket{d}=\cos\d\ket{0}+\sin\d\ket{1},
\quad\quad\quad
\ket{d'}=\sin\d\ket{0}-\cos\d\ket{1},
\end{eqnarray}
and $\a,\b,\g,\d\in(0,\p/2)$.
Hence the UOM becomes
\begin{eqnarray}
\bma
0&0&1&a\\
0&1&a&0\\
0&a&0&1\\
1&1&0&a\\
1&a&1&0\\
1&0&a&1\\
a&0&0&0\\
a&1&1&1\\
a'&a'&a'&a'\\	
\ema
\ra
\bma
0&0&1&d\\
0&1&c&0\\
0&b&0&1\\
1&1&0&d\\
1&b&1&0\\
1&0&c&1\\
a&0&0&0\\
a&1&1&1\\
a'&b'&c'&d'\\	
\ema.	
\end{eqnarray}
Let
\begin{eqnarray}
&&
x_1=\sin\b\sin\g\sin\d
\quad\quad
x_5=\cos\b\cos\g\cos\d,
\\&&
x_2=
\cos\a\cos\g\sin\d,
\quad\quad
x_6=
\sin\a\sin\g\cos\d,
\\&&
x_3=
\cos\a\sin\b\cos\d,
\quad\quad
x_7=
\sin\a\cos\b\sin\d,
\\&&
x_4=
\cos\a\cos\b\sin\g,
\quad\quad
x_8=
\sin\a\sin\b\cos\g,
\end{eqnarray}
and
\begin{eqnarray}
[u_{ij}]:=	
\bma
x_1&x_2&x_3&x_4&-x_5&-x_6&-x_7&-x_8\\
x_2&-x_1&-x_4&x_3&x_6&-x_5&-x_8&x_7\\
x_3&x_4&-x_1&-x_2&x_7&x_8&-x_5&-x_6\\
x_4&-x_3&x_2&-x_1&x_8&-x_7&x_6&-x_5\\
-x_5&-x_6&-x_7&-x_8&-x_1&-x_2&-x_3&-x_4\\
-x_6&x_5&-x_8&x_7&x_2&-x_1&x_4&-x_3\\
x_7&-x_8&-x_5&x_6&-x_3&x_4&x_1&-x_2\\
x_8&x_7&-x_6&-x_5&-x_4&-x_3&x_2&x_1\\
\ema.
\end{eqnarray}
One can verify that $[u_{ij}]$ is an $8\times8$ real unitary matrix. By using the map $\cH_A\otimes\cH_B\ra\bbC^4$ and $\cH_C\otimes\cH_D\ra\bbC^4$, we set $\ket{j,k}:=\ket{2j+k}$ for $j,k=0,1$. So we can define the following 4-qubit pure states as bipartite states in $\bbC^4\otimes\bbC^4$, for $i=1,2,...,8$.
\begin{eqnarray}
\ket{\ps_i}
&&:=u_{i1}\ket{a',0,0,0}	
+
u_{i5}\ket{a',1,1,1}
+
u_{i2}\ket{1,b',1,0}
+
u_{i6}\ket{0,b',0,1}
\notag\\&&
+
u_{i3}\ket{1,0,c',1}
+
u_{i7}\ket{0,1,c',0}
+
u_{i4}\ket{1,1,0,d'}
+
u_{i8}\ket{0,0,1,d'}
\notag\\&&=
u_{i1}\sin\a\ket{00}+
u_{i6}\sin\b\ket{01}+
u_{i8}\sin\d\ket{02}-
u_{i8}\cos\d\ket{03}
\notag\\&&+
u_{i7}\sin\g\ket{10}-
u_{i6}\cos\b\ket{11}-
u_{i7}\cos\g\ket{12}+
u_{i5}\sin\a\ket{13}
\notag\\&&-
u_{i1}\cos\a\ket{20}+
u_{i3}\sin\g\ket{21}+
u_{i2}\sin\b\ket{22}-
u_{i3}\cos\g\ket{23}
\notag\\&&+
u_{i4}\sin\d\ket{30}-
u_{i4}\cos\d\ket{31}-
u_{i2}\cos\b\ket{32}-
u_{i5}\cos\a\ket{33}.
\end{eqnarray}
One can verify that $\ket{\ps_1}=\ket{a',b',c',d'}$. We construct the $4$-qubit PPT entangled state $\r$ of rank seven as follows.
\begin{eqnarray}
\label{eq:rho}
\r &:=&
{1\over7}
(I_{16}
-
\proj{0,0,1,d}-\proj{0,1,c,0}-\proj{0,b,0,1}
\notag\\&-&
\proj{1,1,0,d}-\proj{1,b,1,0}-\proj{1,0,c,1}
\notag\\&-&
\proj{a,0,0,0}-\proj{a,1,1,1}-\proj{a',b',c',d'})
\notag\\&=&
{1\over7}
(\proj{a',0,0,0}
+\proj{1,b',1,0}
+\proj{1,0,c',1}+\proj{1,1,0,d'}
\notag\\&+&
\proj{a',1,1,1}+\proj{0,b',0,1}
+\proj{0,1,c',0}+\proj{0,0,1,d'}
\notag\\&-&
\proj{a',b',c',d'})
\notag\\&=&
{1\over7}\sum^8_{i=2}\proj{\ps_i}.
\end{eqnarray}
It follows from Lemma \ref{le:size9,family11} that $\r_{AB}$ and $\r_{CD}$ both have rank four. Using Proposition \ref{pp:size9}, $\r_{A:B:CD}$ and $\r_{AB:CD}$ are respectively a $2\times2\times4$ and $4\times4$ PPT entangled state of rank seven. This is the first main result of this section.

In the remaining of this section, we investigate the geometric measure of entanglement of $\r$ in \eqref{eq:rho} \cite{wg2003,cxz2010}. For an $n$-partite quantum state $\s$, the measure is defined as
\begin{eqnarray}
\label{eq:gme}
G(\s):=
-\log_2
\max_{a_1,...,a_n}
\bra{a_1,...,a_n}	
\s
\ket{a_1,...,a_n},
\end{eqnarray}
where $\ket{a_1,...,a_n}$ is a normalized product state in $(\bbC^2)^{\otimes n}$. To evaluate $G(\r)$, we assume that
\begin{eqnarray}
\ket{a_j}
=
\bma
\cos\n_j
\\
e^{i\m_j}\sin\n_j
\ema,	
\end{eqnarray}
where the variables $\m_j\in[0,2\p]$ and $\n_j\in[0,\p/2]$ for $j=1,2,3,4$. Using \eqref{eq:abcd}	we have the constant $\a,\b,\g,\d\in(0,\p/2)$ and 
\begin{eqnarray}
\label{eq:grho}
G(\r)
=
-\log_2
\max_{\m_1,\n_1,...,\m_4,\n_4} g(\m_1,\m_2,\m_3,\m_4,\n_1,\n_2,\n_3,\n_4)	
\end{eqnarray}
where
\begin{eqnarray}
\label{eq:gabcd1}
&&
g(\m_1,\m_2,\m_3,\m_4,\n_1,\n_2,\n_3,\n_4)
\notag\\&=&
\bra{a_1,a_2,a_3,a_4}	
\r
\ket{a_1,a_2,a_3,a_4}	
\notag\\
\label{eq:gabcd2}
&=&
{1\over7}\bigg(
\notag\\&&
(\sin^2\a\cos^2\n_1+\cos^2\a\sin^2\n_1
-2\cos\m_1\sin\a\cos\n_1\cos\a\sin\n_1)
(\cos^2\n_2\cos^2\n_3\cos^2\n_4
+\sin^2\n_2\sin^2\n_3\sin^2\n_4)
\notag\\&+&
(\sin^2\b\cos^2\n_2+\cos^2\b\sin^2\n_2
-2\cos\m_2\sin\b\cos\n_2\cos\b\sin\n_2)
(\sin^2\n_1\sin^2\n_3\cos^2\n_4
+\cos^2\n_1\cos^2\n_3\sin^2\n_4)
\notag\\&+&
(\sin^2\g\cos^2\n_3+\cos^2\g\sin^2\n_3
-2\cos\m_3\sin\g\cos\n_3\cos\g\sin\n_3)
(\sin^2\n_1\cos^2\n_2\sin^2\n_4
+\cos^2\n_1\sin^2\n_2\cos^2\n_4)
\notag\\&+&
(\sin^2\d\cos^2\n_4+\cos^2\d\sin^2\n_4
-2\cos\m_4\sin\d\cos\n_4\cos\d\sin\n_4)(\sin^2\n_1\sin^2\n_2\cos^2\n_3
+\cos^2\n_1\cos^2\n_2\sin^2\n_3)
\notag\\&-&
(\sin^2\a\cos^2\n_1+\cos^2\a\sin^2\n_1
-2\cos\m_1\sin\a\cos\n_1\cos\a\sin\n_1)
\notag\\&&
(\sin^2\b\cos^2\n_2+\cos^2\b\sin^2\n_2
-2\cos\m_2\sin\b\cos\n_2\cos\b\sin\n_2)
\notag\\&&
(\sin^2\g\cos^2\n_3+\cos^2\g\sin^2\n_3
-2\cos\m_3\sin\g\cos\n_3\cos\g\sin\n_3)
\notag\\&&
(\sin^2\d\cos^2\n_4+\cos^2\d\sin^2\n_4
-2\cos\m_4\sin\d\cos\n_4\cos\d\sin\n_4)
\bigg).
\end{eqnarray}
Since $g(\m_1,\m_2,\m_3,\m_4,\n_1,\n_2,\n_3,\n_4)$ is a linear function with $\cos\m_j$, its maximum is achieved when $\cos\m_j=1$ or $-1$ for $j=1,2,3,4$. In this case we have $\cos\m_i\sin\n_j=\sin(\n_j\cos\m_i)$. Recall that $\n_j\in[0,{\p\over2}]$. Letting $\n_j\cos\m_i=\l_j$ we can assume that $\l_j\in[-{\p\over2},{\p\over2}]$ for computing $G(\r)$.

To demonstrate out method, we consider a special $\r$ in \eqref{eq:rho} by choosing
the constant $\a=\b=\g=\d={\p\over4}$. Using the above conditions we can obtain 
\begin{eqnarray}
\label{eq:f1g1+f2g2}
&&
g(\m_1,\m_2,\m_3,\m_4,\n_1,\n_2,\n_3,\n_4)
:=
h(\l_1,\l_2,\l_3,\l_4)
=
{1\over7}(f_1g_1+f_2g_2+f_3g_3+f_4g_4
-f_1f_2f_3f_4),
\end{eqnarray}
and the eight functions
\begin{eqnarray}
\label{eq:f1234def}
&&
f_1=f(\l_1):={1\over2}
-{1\over2}\sin2\l_1,
\notag\\&&
f_2=f(\l_2):={1\over2}
-{1\over2}\sin2\l_2,
\notag\\&&
f_3=f(\l_3):=
{1\over2}
-{1\over2}\sin2\l_3,
\notag\\&&
f_4=f(\l_4):={1\over2}-{1\over2}\sin2\l_4,
\notag\\&&
g_1=g_1(\l_2,\l_3,\l_4):=
\cos^2\l_2\cos^2\l_3\cos^2\l_4
+\sin^2\l_2\sin^2\l_3\sin^2\l_4,	
\notag\\&&
g_2=g_2(\l_1,\l_3,\l_4):=\sin^2\l_1\sin^2\l_3\cos^2\l_4
+\cos^2\l_1\cos^2\l_3\sin^2\l_4,
\notag\\&&
g_3=g_3(\l_1,\l_2,\l_4):=
\sin^2\l_1\sin^2\l_4\cos^2\l_2
+\cos^2\l_1\cos^2\l_4\sin^2\l_2,
\notag\\&&
g_4=g_4(\l_1,\l_2,\l_3):=\sin^2\l_1\sin^2\l_2\cos^2\l_3
+\cos^2\l_1\cos^2\l_2\sin^2\l_3.
\end{eqnarray}
By checking the necessary conditions 
\begin{eqnarray}
\pp{h(\l_1,\l_2,\l_3,\l_4)}{\l_2}=	
\pp{h(\l_1,\l_2,\l_3,\l_4)}{\l_3}=
\pp{h(\l_1,\l_2,\l_3,\l_4)}{\l_4}=0,
\end{eqnarray} 
we obtain that $\l_2=\l_3=\l_4$. Since $h(\l_1,\l_2,\l_2,\l_2)
$ is a linear function with $\sin 2\l_1$, its maximum is achieved when $\sin 2\l_1=-1$ or $1$. When $\sin 2\l_1=1$, the extremum of $h(\l_1,\l_2,\l_2,\l_2)$ is $0$ and $\frac{1}{126}$. When $\sin 2\l_1=-1$, the extremum of $h(\l_1,\l_2,\l_2,\l_2)$ is  ${3\over28}\sqrt{{3\over2}}$. One can derive that the maximum of $h(\l_1,\l_2,\l_3,\l_4)$ is ${3\over28}\sqrt{{3\over2}}\approx0.131$, when $\l_1=-{\p\over4}$ and $\l_2={1\over2}\arcsin{\sqrt6-2\over2}$. It follows from \eqref{eq:grho} and \eqref{eq:f1g1+f2g2} that
\begin{theorem}
\label{thm:gme}
For the 4-qubit PPT entangled state $\r$ in \eqref{eq:rho} with $\a=\b=\g=\d={\p\over4}$, its geometric measure of entanglement is $G(\r)=-\log_2 {3\over28}\sqrt{{3\over2}}\approx2.93$ ebits.	
\end{theorem}
This is the second main result of this section. We have shown  below \eqref{eq:rho} that both $\r_{A:B:CD}$ and $\r_{AB:CD}$ are PPT entangled states. Using the definition of UPBs we have
\begin{eqnarray}
G(\r)
\ge 	
G(\r_{A:B:CD})
\ge 
G(\r_{AB:CD}).
\end{eqnarray}
So Theorem \ref{thm:gme} gives an upper bound of the geometric measure of entanglement of both $\r_{A:B:CD}$ and $\r_{AB:CD}$. By varying the constants $\a,\b,\g,\d$ in \eqref{eq:abcd}, one can similarly investigate the entanglement of more states $\r$ in \eqref{eq:rho}.

\section{Conclusions}
\label{sec:con}

We have applied the classification of 4-qubit UPBs to construct more UPBs. We have shown that there is only one UPB of size $6$ in $\bbC^2\ox\bbC^2\ox\bbC^4$, no UPB of size $7$, six UPBs of size $9$ in $\bbC^2\ox\bbC^2\ox\bbC^4$ and three UPBs of size $9$ in $\bbC^4\ox\bbC^4$. As an application of our results on UPBs of size $9$, we have constructed a family of PPT entangled states $\r$ of rank seven for the systems of $\bbC^2\ox\bbC^2\ox\bbC^2\ox\bbC^2$, $\bbC^2\ox\bbC^2\ox\bbC^4$ and $\bbC^4\ox\bbC^4$ at the same time. Furthermore we have worked out the entanglement of a constant 4-qubit $\r$ using the geometric measure of entanglement. 

Our results have provided better understanding of UPBs in $\bbC^4\ox\bbC^4$ and $\bbC^2\ox\bbC^2\ox\bbC^4$. The next step is to investigate the set $\cU_{2,2,4}^8$. It includes as a subset the set $\cU_{4,4}^8$ we have found in this paper. Primary investigation shows that the inclusion is strict, namely there exist UPBs in the coarse graining $\bbC^2\ox\bbC^2\ox\bbC^4$, though they are not UPBs in the coarse graining $\bbC^2\ox\bbC^4$.

\section*{Acknowledgments}

This work was supported by the NNSF of China (Grant No. 11871089), and the Fundamental Research Funds for the Central Universities (Grant Nos. KG12040501, ZG216S1810 and ZG226S18C1).

\bibliographystyle{unsrt}

\bibliography{channelcontrol}

\end{document}